\documentclass[conference]{IEEEtran}
\IEEEoverridecommandlockouts
\usepackage{cite}
\usepackage{amsmath,amssymb,amsfonts}
\usepackage{algorithmic}
\usepackage{graphicx}
\usepackage{textcomp}
\usepackage{xcolor}
\def\BibTeX{{\rm B\kern-.05em{\sc i\kern-.025em b}\kern-.08em
    T\kern-.1667em\lower.7ex\hbox{E}\kern-.125emX}}

\usepackage{amsthm}

\newtheorem{lem}{Lemma}

\newtheorem{thm}{Theorem}

\begin{document}

\title{Stochastic Geometry Analysis of Spectrum Sharing Among Multiple Seller and Buyer Mobile Operators}

\author{\IEEEauthorblockN{Elaheh~Ataeebojd$^{\dag}$, Mehdi Rasti$^{\dag}$, Hossein Pedram$^{\dag}$, and Pedro H. J. Nardelli$^{\S}$
	\thanks{This work is supported by the Academy of Finland: (a) ee-IoT n.319009, (b) EnergyNet n.321265/n.328869, and (c) FIREMAN n.326270/CHISTERA-17-BDSI-003; and by JAES Foundation via STREAM project.}
}
$^{\dag}$Department of Computer Engineering, Amirkabir University of Technology, Tehran, Iran\\
$^{\S}$ Lappeenranta-Lahti University of Technology, Lappeenranta, Finland\\
}
\maketitle

\begin{abstract}
Sharing the licensed frequency spectrum among multiple mobile network operators (MNOs) is a promising approach to improve licensed spectrum utilization. In this paper, we model and analyze a non-orthogonal spectrum sharing system consisting of multiple seller and multiple buyer MNOs where buyer MNOs lease several licensed sub-bands from different seller MNOs. All base stations (BSs) owned by a buyer MNO can also utilize various licensed sub-bands simultaneously, which are also used by other buyer MNOs. To reduce the interference that a buyer MNO imposes on one seller MNO sharing its licensed sub-band, this buyer MNO has a limitation on the maximum interference caused to the corresponding seller MNO's users. We assume each MNO owns its BSs and users whose locations are modeled as two independent homogeneous Poisson point processes. Applying stochastic geometry, we derive expressions for the downlink signal-to-interference-plus-noise ratio coverage probability and the average rate of both seller and buyer networks. The numerical results validate our analysis with simulation and illustrate the effect of the maximum interference threshold on the total sum-rate of the network.
\end{abstract}


\section{Introduction}
Due to the scarcity of frequency spectrum and increasing demand for various mobile services \cite{Cisco}, mobile network operators (MNOs) encounter many challenges in providing efficient ways to use resources more efficiently to serve their users in 5G and beyond. In recent studies, spectrum sharing is considered as a robust solution in 5G and beyond to utilize the existing resources more efficiently \cite{Gui2020}. In spectrum sharing, different MNOs, named buyer MNOs, are allowed to use (or lease) the frequency sub-bands of other MNOs, called seller MNOs, to afford mobile data services to their users \cite{Umar2017}.

In \cite{Joshi2017} -- \cite{Asaduzzaman2018}, various aspects and performance of spectrum sharing among MNOs are studied in a cellular setting. These studies consider a deterministic model in which the location of base stations (BSs) and user equipment (UE) are known, resulting in unwieldy problem formulation when the number of BSs and UEs is large. Hence, a stochastic geometry model is adopted as a tractable approach to model and analyze various cellular wireless systems owned by different MNOs as the case in \cite{Andrews2016} -- \cite{Gupta2016} instead of a deterministic model. For instance, in \cite{Gupta2016}, the performance of a system with a single buyer and a single seller MNO is investigated, and performance metrics such as signal-to-interference-plus-noise ratio (SINR) coverage probability and per-user average rate are derived. Besides, in \cite{Gupta2016}, a power control strategy is adopted based on the distance between each buyer BS and the nearest seller UE to it. In \cite{Sanguanpuak2017Conf}, a stochastic geometry framework is presented to analyze a non-orthogonal spectrum sharing system in an indoor small cell network considering a single buyer MNO and multiple seller MNOs. Spectrum sharing among multiple MNOs is also studied in \cite{Andrews2016}, where each MNO owns its BSs and UEs distributed based on independent Poisson Point Process (PPP). 

A significant limitation of prior work \cite{Joshi2017} -- \cite{Sanguanpuak2017Conf} is that in sharing the frequency spectrum between different MNOs, there is no power control over the transmit power of the BSs, especially the BSs of buyer MNOs. Consequently, the interference imposed on UEs is unpredictable, leading to a reduction in the overall performance of the networks. The other limitation is the assumption of a single buyer MNO in \cite{Sanguanpuak2017Conf} -- \cite{Gupta2016}. To study the impact of sharing sub-bands of one seller MNO with several buyer MNOs and the effect of simultaneous leasing sub-bands from multiple seller MNOs by each buyer MNO (opposing with \cite{Andrews2016}), a more general framework containing multiple seller and multiple buyer MNOs is needed.

In this paper, to address the issues mentioned above, we consider a non-orthogonal spectrum sharing system where multiple buyer MNOs lease several licensed sub-bands from multiple seller MNOs. In this scenario, all BSs of each buyer MNO can utilize several sub-bands. Also, each seller MNO allows multiple buyer MNOs to utilize each of its sub-band simultaneously. Moreover, inspired by \cite{Yan2018}, we introduce a distribution function for the downlink transmit power of buyer MNOs so that the interference imposed on seller MNOs' UEs does not violate a tolerable interference threshold. Using a stochastic geometry approach, we analyze the SINR coverage probability and the average rate considering a large-scale cellular network. Our contributions in this paper are summarized as follows:
\begin{itemize}
	\item { 
		A more general framework is proposed containing multiple seller and multiple buyer MNOs to study the impact of sharing sub-bands between seller and buyer MNOs.
	}
	\item {
		Employing stochastic geometry, we analytically obtain the SINR coverage probability and the average rate as performance metrics under random parameters such as channel gains, distance of UEs from BSs, as well as the transmit power of buyer BSs. To do this, we introduce a power control strategy to justify the transmit power of buyer MNOs' BSs on the licensed sub-bands of the seller MNOs.
	}
	\item {
		Numerical results validate our analytical results with simulation and show the performance of our adopted power control strategy.
	}
\end{itemize}

The rest of this paper is organized as follows: The system model and our assumptions are explained in Section \ref{Sys}. Employing stochastic geometry, we analyze the probability of SINR coverage and the average rate for both seller and buyer MNOs in Section \ref{Ana}. In Section \ref{Res}, the numerical results are presented. Finally, the conclusion is summarized in Section \ref{Con}.
\vspace{-0.7 em}
\section{System Model and Assumptions}\label{Sys}
\subsection{System Model}
Consider a downlink cellular wireless network with a set of seller MNOs $\mathcal{S}$ lease their frequency sub-bands to a set of buyer MNOs $\mathcal{B}$ (see Fig. \ref{fig:Spect}). The set of all MNOs is denoted by $\mathcal{O}=\mathcal{S}\cup\mathcal{B}$. We assume that seller MNO $s\in\mathcal{S}$ owns a license for an orthogonal spectrum divided into $L_{s}$ licensed sub-bands. Moreover, the set of licensed sub-bands for seller MNO $s\in\mathcal{S}$ is represented by $\mathcal{L}_{s}$ $(|\mathcal{L}_{s}| = L_{s})$. Additionally, there is no interference among the licensed sub-bands of different seller MNOs. Each MNO owns an independent network consisting of its own BSs and UEs. The location of BSs and UEs for MNO $k \in\mathcal{O}$ are modeled by two independent PPPs $\Phi_{k}$ and $\Psi_{k}$ with intensity $\lambda_{k}$ and $\mu_{k}$, respectively. Also, the set of UEs and BSs owned by MNO $k\in\mathcal{O}$ are given by $\mathcal{U}_{k}$ and $\mathcal{F}_{k}$, respectively. Each BS of seller MNOs transmits with a fixed power denoted by $P^{\text{S}}$ on their sub-bands. Nonetheless, each BS of buyer MNOs is limited to transmit at a certain power $P^{\text{B}}_{l_{s}}$ on sub-band $l_{s}\in\mathcal{L}_{s}$ of seller MNO $s$ to ensure that the interference imposed on each UE of seller MNO $s$ is below the maximum interference threshold $\zeta_{s}$. Furthermore, it is assumed that a UE associated with a BS of its MNO, providing the maximum average received power. We perform the analysis based on a typical seller UE (SUE) of a seller MNO and a typical buyer UE (BUE) of a buyer MNO for downlink cellular communication. We use a general power-law path loss model with path loss exponent $\alpha > 2$. Besides, we assume all channel gains are modeled by Rayleigh fading, which means that all fading variables are exponentially distributed with unit mean.
\begin{figure}[t!]
	\center{\includegraphics[width=9 cm,height=6 cm]
		{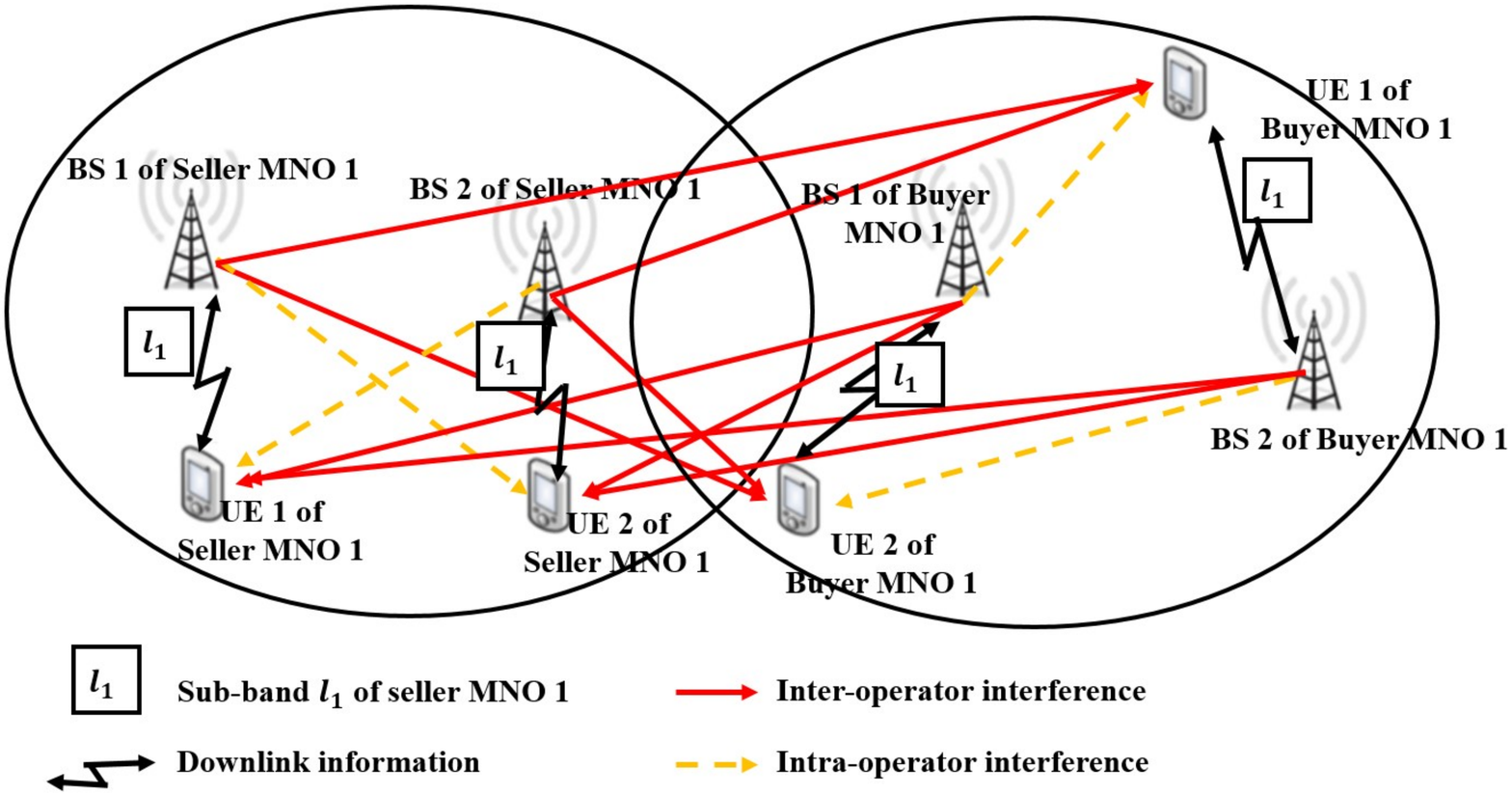}
		\caption{A cellular network with one seller and one buyer MNO.\label{fig:Spect}}
	}
\end{figure}
\vspace{-7pt}
\subsection{License Sharing Group}
Seller MNO $s$ can share each of its licensed sub-bands, e.g., $l_{s}\in\mathcal{L}_{s}$, with multiple buyer MNOs at the same time (see Fig. \ref{fig:Spect}). Let us denote all MNOs sharing licensed sub-band $l_{s}\in\mathcal{L}_{s}$ by $Q_{l_{s}}$, where seller MNO $s$ is surely belong to $Q_{l_{s}}$. Accordingly, in each sub-band $l_{s}\in\mathcal{L}_{s}$, a typical BUE of buyer MNO $b$ experiences interference from the BSs of its MNO, the BSs of seller MNO $s$, and the BSs of the other buyer MNOs sharing sub-band $l_{s}\in\mathcal{L}_{s}$ simultaneously. Similarly, in each sub-band $l_{s}\in\mathcal{L}_{s}$, a typical SUE of seller MNO $s$ will experience interference from the BSs of its MNO and the BSs of buyer MNOs sharing sub-band $l_{s}\in\mathcal{L}_{s}$ at the same time. In both cases, the set of interfering BSs is owned by MNOs, members of the sharing group $Q_{l_{s}}$.
\subsection{SINR Model}
We assume that a typical SUE of seller MNO $s$ is located at the origin. Thus, the downlink received SINR at the typical SUE of seller MNO $s\in\mathcal{S}$ on sub-band $l_{s}\in\mathcal{L}_{s}$ is determined by
\begin{equation}\label{SINR_n}
\gamma^{\text{S}}_{s,l_{s}}=\frac{P^{\text{S}}h_{0,l_{s}}^{\text{S}}{x_{0}^{\text{S}}}^{-\alpha}}{I_{s,l_{s}}^{\text{S}}+\sigma^{2}},
\end{equation}
where $x_{0}^{\text{S}}$ is the distance between this typical SUE and its associated BS, $h_{0,l_{s}}^{\text{S}}$ represents the channel gain between this typical SUE and its associated BS on sub-band $l_{s}\in\mathcal{L}_{s}$, ${\sigma}^2$ denotes the noise power, and $I_{s,l_{s}}^{\text{S}}$ is the interference experienced by this typical SUE of seller MNO $s$ on sub-band $l_{s}\in\mathcal{L}_{s}$ from all interfering BSs. It is calculated as $I_{s,l_{s}}^{\text{S}}=\sum_{f\in\mathcal{F}_{s}\setminus\{0\}} P^{\text{S}} h_{f,l_{s}}^{\text{S}} \allowbreak {x_{f}^{\text{S}}}^{-\alpha} + \sum_{b \in Q_{l_{s}}\setminus{\{s\}}} \sum_{f\in\mathcal{F}_{b}} P^{\text{B}}_{l_{s}} h_{f,l_{s}}^{\text{B}} {x_{f}^{\text{B}}}^{-\alpha}$, where $h_{f,l_{s}}^{\text{S}}$ is the channel gain between this typical SUE and BS $f$ of the corresponding seller MNO on sub-band $l_{s}\in\mathcal{L}_{s}$, $h_{f,l_{s}}^{\text{B}}$ is the channel gain between this typical SUE and BS $f$ of any buyer MNO on sub-band $l_{s}\in\mathcal{L}_{s}$, $x_{f}^{\text{S}}$ is the distance between this typical SUE and BS $f$ of the corresponding seller MNO, and $x_{f}^{\text{B}}$ is the distance between this typical SUE and BS $f$ of any buyer MNO.

Likewise, we assume that a typical BUE of buyer MNO $b$ is located at the origin. Therefore, the downlink received SINR at the typical BUE of buyer MNO $b$ on shared sub-band $l_{s}\in \mathcal{L}_{s}$ is expressed as
\begin{equation} \label{SINR_m}
\gamma^{\text{B}}_{b,l_{s}}=\frac{P^{\text{B}}_{l_{s}}h_{0,l_{s}}^{\text{B}}{x_{0}^{\text{B}}}^{-\alpha}}{I_{b,l_{s}}^{\text{B}}+\sigma^{2}},
\end{equation}
where $x_{0}^{\text{B}}$ represents the distance between this typical BUE and its associated BS, $h_{0,l_{s}}^{\text{B}}$ denotes the channel gain between this typical BUE and its associated BS on sub-band $l_{s}\in\mathcal{L}_{s}$, and $I_{b,l_{s}}^{\text{B}}$ is the interference experienced by this typical BUE of buyer MNO $b$ on sub-band $l_{s}\in\mathcal{L}_{s}$ from all interfering BSs. It is calculated as $I_{b,l_{s}}^{\text{B}} = \sum_{f\in\mathcal{F}_{b}\setminus\{0\}} P^{\text{B}}_{l_{s}} h_{f,l_{s}}^{\text{B}} {x_{f}^{\text{B}}}^{-\alpha} + \sum_{f \in \mathcal{F}_{s}}\allowbreak P^{\text{S}} h_{f,l_{s}}^{\text{S}} {x_{f}^{\text{S}}}^{-\alpha} + \sum_{k \in Q_{l_{s}}\setminus\{s,b\}} \sum_{f\in\mathcal{F}_{k}} P^{\text{B}}_{l_{s}} h_{f,l_{s}}^{\text{B}}\allowbreak {x_{f}^{\text{B}}}^{-\alpha}$.
\subsection{Power Control Strategy} \label{Power_Strategy}
Inspired by the approach in \cite{Yan2018}, each buyer BS is limited to transmit at a certain power $P^{\text{B}}_{l_{s}}$ on sub-band $l_{s}\in\mathcal{L}_{s}$ to ensure that the interference imposed on each UE of seller MNO $s$ is below the maximum interference threshold $\zeta_{s}$. That is, given the maximum interference threshold $\zeta_{s}$, for any SUE $i\in \mathcal{U}_{s}$ and an arbitrary buyer BS which are located at ${\rm{y}}_{i}^{\text{S}}$ and ${\rm{x}}^{\text{B}}$, respectively, we have $P_{l_{s}}^{\text{B}}h_{i,0,l_{s}}{||{\rm{x}}^{\text{B}}-{\rm{y}}_{i}^{\text{S}}||}^{-\alpha}\leq \zeta_{s}$, where $h_{i,0,l_{s}}$ is the channel gain between UE $i$ and this arbitrary BS on sub-band $l_{s}\in\mathcal{L}_{s}$. Accordingly, given the maximum interference threshold $\zeta_{s}$, it is sufficient to satisfy $\max_{\substack{{\rm{y}}_{i}^{\text{S}}\in\Psi_{s}\\\forall i\in\mathcal{U}_{s}}} P_{l_{s}}^{\text{B}}h_{i,0,l_{s}}{||{\rm{x}}^{\text{B}}-{\rm{y}}_{i}^{\text{S}}||}^{-\alpha}\leq \zeta_{s}$. By doing this, the interference constraint is surely satisfied for all other UEs of seller MNO $s$. Let us denote the largest interference channel gain on sub-band $l_{s}$ associated with each buyer BS as ${H}_{l_{s}}^{\text{B}}$. It is defined as ${H}_{l_{s}}^{\text{B}} = \max_{\substack{{\rm{y}}_{i}^{\text{S}}\in\Psi_{s}\\\forall i\in\mathcal{U}_{s}}} h_{i,0,l_{s}}{||{\rm{x}}^{\text{B}}-{\rm{y}}_{i}^{\text{S}}||}^{-\alpha}$. If each buyer BS transmits with the maximum allowable power on licensed sub-band $l_{s}\in\mathcal{L}_{s}$, $P^{\text{B}}_{l_{s}}$ is calculated by
\begin{equation} \label{Power}
P^{\text{B}}_{l_{s}}=\zeta_{s}/H^{\text{B}}_{l_{s}}.
\end{equation}
\vspace{-1.1 em}
\section{Analysis of Performance Metrics}\label{Ana}
In this section, different performance metrics are analyzed. In subsections \ref{S_A} and \ref{S_B}, the coverage probability and average rate, as performance metrics, for UEs are obtained, respectively. According to \eqref{SINR_n} and \eqref{SINR_m}, SINR is dependent on $P_{l_{s}}^{\text{B}}$ which is a random variable, therefore the distribution of $P_{l_{s}}^{\text{B}}$ is needed to drive the coverage probability and average rate of both BUEs and SUEs. Since $P_{l_{s}}^{\text{B}}$ is dependent on $H_{l_{s}}^{\text{B}}$ based on \eqref{Power}, we first introduce the distribution of $H_{l_{s}}^{\text{B}}$ and then the distribution of $P_{l_{s}}^{\text{B}}$ in the following Lemmas \cite{Yan2018}.
\begin{lem} \label{lemma_1}
	The cumulative distribution function (CDF) and the probability density function (PDF) of $H^{\text{B}}_{l_{s}}$ are expressed respectively as 
	\begin{equation}
	F_{H^{\rm{B}}_{l_{s}}}(z)=\exp\left(\frac{-\pi\mu_{s}}{z^{2/\alpha}}\Gamma(1+\frac{2}{\alpha})\right)\label{CDF_H},
	\end{equation}
	and
	\begin{equation}
	f_{H^{\rm{B}}_{l_{s}}}(z)= \frac{2\pi\mu_{s}\Gamma(1+\frac{2}{\alpha})}{\alpha z^{1+\frac{2}{\alpha}}} \exp\left(\frac{-\pi\mu_{s}}{z^{2/\alpha}}\Gamma(1+\frac{2}{\alpha})\right), \label{PDF_H}
	\end{equation}
	where $\Gamma(\kappa)$ is the complete Gamma function calculated as $\Gamma(\kappa)=\int_{0}^{+\infty}{t^{\kappa-1}e^{-t}{\rm{d}}t}$.
\end{lem}
\begin{lem} \label{lemma_2}
	The CDF and PDF of $P^{\text{B}}_{l_{s}}$ are expressed respectively as
	\begin{equation}
	F_{P^{\rm{B}}_{l_{s}}}(z)=1-\exp\left(\frac{-\pi\mu_{s}z^{2/\alpha}}{\zeta_{s}^{2/\alpha}}\Gamma(1+\frac{2}{\alpha})\right)\label{CDF_P},
	\end{equation}
	and
	\begin{equation}
	f_{P^{\rm{B}}_{l_{s}}}(z)= \frac{2\pi\mu_{s}z^{2/\alpha-1}\Gamma(1+\frac{2}{\alpha})}{\alpha \zeta_{s}^{\frac{2}{\alpha}}} \exp\left(\frac{-\pi\mu_{s}z^{2/\alpha}}{\zeta_{s}^{2/\alpha}}\Gamma(1+\frac{2}{\alpha})\right). \label{PDF_P}
	\end{equation}
\end{lem}
Now, having the distribution of $P^{\text{B}}_{l_{s}}$, we aim to analyze the coverage probability and average rate for both BUEs and SUEs. To do this, without loss of generality, we analyze the coverage probability and average rate for a typical SUE of a seller MNO and a typical BUE of a buyer MNO considering a network in which multiple seller MNOs lease their sub-bands to multiple buyer MNOs. Similarly, the coverage probability and average rate for other BUEs and SUEs are analyzed.
\subsection{Coverage probability for buyer and seller MNOs}\label{S_A}
For a typical UE of MNO $k\in\mathcal{O}$ and a threshold $\beta$, the SINR coverage probability on sub-band $l_{s}\in\mathcal{L}_{s}$ is defined as the probability that SINR at this typical UE on sub-band $l_{s}\in\mathcal{L}_{s}$ is above $\beta$, that is
\begin{equation}\label{CP}
\mathcal{C}_{k,l_{s}}^{\mathcal{X}} = \mathbb{E}_{P^{\text{B}}_{l_{s}}}\left[\mathbb{P}\left[\gamma_{k,l_{s}}^{\mathcal{X}} > \beta\right]\right],
\end{equation}
where \eqref{CP} holds for any arbitrary UE owned by MNO $\mathcal{X}\in\{\text{S}, \text{B}\}$, where $\text{S}$ and $\text{B}$ represent the seller and buyer MNOs, respectively. In the following Theorems, the coverage probability for a typical BUE and a typical SUE is computed, respectively.
\begin{thm} \label{thm_1}
	The coverage probability for a typical BUE of buyer MNO $b$ on sub-band $l_{s}\in\mathcal{L}_{s}$ is
	\begin{equation}
	\begin{aligned}
	\mathcal{C}_{b,l_{s}}^{\rm{B}}&=\mathbb{E}_{P^{\rm{B}}_{l_{s}}} \left[\mathbb{P}\left[{\gamma}_{b,l_{s}}^{\rm{B}} > \beta\right]\right]\\
	&=\int_{0}^{+\infty} \pi\lambda_{b}\mathbb{E}_{P_{l_{s}}^{\rm{B}}}\left[{P_{l_{s}}^{\rm{B}}}^{2/\alpha}\right]\exp\left(-\beta\sigma^{2}z^{\alpha/2}\right)\\
	&\times\exp\left(-\pi\lambda_{b}\mathbb{E}_{P_{l_{s}}^{\rm{B}}}\left[{P_{l_{s}}^{\rm{B}}}^{2/\alpha}\right]z\right) \mathcal{L}_{I_{b,l_{s}}^{\rm{B}}}\left(\beta z^{\alpha/2}\right) {\rm{d}}z,
	\end{aligned}
	\end{equation}
	where $\rho(\alpha,\beta)=\int_{{\beta}^{-2/\alpha}}^{+\infty}\frac{1}{1+\nu^{\alpha/2}} {\rm{d}}\nu$, $\rho(\alpha,\infty)=\int_{0}^{+\infty}\frac{1}{1+\nu^{\alpha/2}}{\rm{d}}\nu$, and $\mathcal{L}_{I_{b,l_{s}}^{\rm{B}}}$ is the Laplace transform of the interference $I_{b,l_{s}}^{\text{B}}$. $\mathcal{L}_{I_{b,l_{s}}^{\rm{B}}}(t)$ is defined as
	\begin{equation}
	\begin{aligned}
	\mathcal{L}_{I_{b,l_{s}}^{\rm{B}}}(\kappa)&=\exp\left(\pi\lambda_{b}\mathbb{E}_{P_{l_{s}}^{\rm{B}}}\left[{P_{l_{s}}^{\rm{B}}}^{2/\alpha}\right]\kappa^{2/\alpha}\rho(\alpha,\beta)\right)\\
	&\times\exp\left(\pi\lambda_{s}{P^{\rm{S}}_{}}\,^{2/\alpha}\kappa^{2/\alpha}\rho(\alpha,\infty)\right)\\
	&\times\exp\left(\pi\lambda_{b,l_{s}}\mathbb{E}_{P_{l_{s}}^{\rm{B}}}\left[{P_{l_{s}}^{\rm{B}}}^{2/\alpha}\right]\kappa^{2/\alpha}\rho(\alpha,\infty)\right),
	\end{aligned}
	\end{equation}
	where $\lambda_{b,l_{s}}=\sum_{k \in Q_{l_{s}}\setminus\{s,b\}}\lambda_{k}$.
\end{thm}
\begin{proof}
	The proof is similar to other works done in PPP,
	e.g. \cite[Lemma 1]{Gupta2016}, (omitted due to page limitation).
\end{proof}
In the following Theorem, the coverage probability for a typical SUE is given.
\begin{thm} \label{thm_2}
	The coverage probability for a typical SUE of seller MNO $s$ on sub-band $l_{s}\in\mathcal{L}_{s}$ is
	\begin{equation}
	\begin{aligned}
	&\mathcal{C}_{s,l_{s}}^{\rm{S}}=\mathbb{E}_{P^{\text{B}}_{l_{s}}} \left[\mathbb{P}\left[{\gamma}_{s,l_{s}}^{\rm{S}} > \beta\right]\right] = \int_{0}^{+\infty} \pi\lambda_{s}{P^{\rm{S}}}\,^{2/\alpha}\\
	&\times\exp\left(-\beta\sigma^{2}z^{\alpha/2}\right)\exp\left(-\pi\lambda_{s}{P^{\rm{S}}}\,^{2/\alpha}z\right) \mathcal{L}_{I_{s,l_{s}}^{\rm{S}}}\left(\beta z^{\alpha/2}\right) {\rm{d}}z,
	\end{aligned}
	\end{equation}
	where $\mathcal{L}_{I_{s,l_{s}}^{\rm{S}}}(t)$ is the Laplace transform of the interference $I_{s,l_{s}}^{\rm{S}}$, that is
	\begin{equation}
	\begin{aligned}
	\mathcal{L}_{I_{s,l_{s}}^{\rm{S}}}(\kappa)&=\exp\left(\pi\lambda_{s}{P^{\rm{S}}_{}}\,^{2/\alpha}\kappa^{2/\alpha}\rho(\alpha,\beta)\right)\\
	&\times\exp\left(\pi\lambda_{s,l_{s}}\mathbb{E}_{P_{l_{s}}^{\rm{B}}}\left[{P_{l_{s}}^{\rm{B}}}^{2/\alpha}\right]\kappa^{2/\alpha}\rho(\alpha,\infty)\right),
	\end{aligned}
	\end{equation}
	where $\lambda_{s,l_{s}}=\sum_{b \in Q_{l_{s}}\setminus\{s\}}\lambda_{b}$.
\end{thm}
\begin{proof}
	The proof is similar to \cite[Lemma 4]{Gupta2016}.
\end{proof}
\subsection{Average rate for buyer and seller MNOs}\label{S_B}
The average rate of a typical BUE of MNO $b$ on all shared sub-bands is given by
\begin{equation}\label{R_m}
\mathcal{R}_{b}^{\text{B}} = \sum_{l_{s}:\,b\in Q_{l_{s}}} \mathbb{E}_{P^{\text{B}}_{l_{s}}}\left[\mathbb{E}_{\gamma_{b,l_{s}}^{\text{B}}}\left[\ln\left(1 + \gamma_{b,l_{s}}^{\text{B}}\right)\right]\right],
\end{equation}
where the sum is over all licensed sub-bands which are leased by buyer MNO $b$. Now, $\mathcal{R}_{b}^{\text{B}}$ is presented in the following Theorem.
\begin{thm} \label{thm_3}
	The average rate for a typical BUE of buyer MNO $b$ is
	\begin{equation} \label{r_m}
	\begin{aligned}
	\mathcal{R}^{\rm{B}}_{b}&=\sum_{l_{s}:b\in Q_{l_{s}}} \mathbb{E}_{P^{\text{B}}_{l_{s}}} \left[\mathbb{E}_{{\gamma}_{b,l_{s}}^{\rm{B}}}\left[\ln\left(1 + {\gamma}_{b,l_{s}}^{\rm{B}}\right)\right] \right]\\
	&=\sum_{l_{s}:b\in Q_{l_{s}}} \mathbb{E}_{P^{\rm{B}}_{l_{s}}} \left[\int_{0}^{\infty}\mathbb{P}\left[\ln\left(1 + {\gamma}_{b,l_{s}}^{\rm{B}}\right) > t\right] {\rm{d}} t\right]\\
	&=\sum_{l_{s}:b\in Q_{l_{s}}} \int_{0}^{\infty}\int_{0}^{+\infty} \pi\lambda_{b}\mathbb{E}_{P_{l_{s}}^{\rm{B}}}\left[{P_{l_{s}}^{\rm{B}}}^{2/\alpha}\right]\\
	&\times\exp\left(-(e^{t}-1)\sigma^{2}z^{\alpha/2}\right)\exp\left(-\pi\lambda_{b}\mathbb{E}_{P_{l_{s}}^{\rm{B}}}\left[{P_{l_{s}}^{\rm{B}}}^{2/\alpha}\right]z\right)\\
	&\times\mathcal{L}_{I_{b,l_{s}}^{\rm{B}}}\left((e^{t}-1)z^{\alpha/2}\right) {\rm{d}}z \;{\rm{d}}t.
	\end{aligned}
	\end{equation}
\end{thm}
The average rate of a typical SUE of MNO $s$ on all its sub-bands is expressed as
\begin{equation}\label{R_n}
\mathcal{R}_{s}^{\text{S}} = \sum_{l_{s}\in \mathcal{L}_{s}} \mathbb{E}_{P^{\text{B}}_{l_{s}}}\left[\mathbb{E}_{\gamma_{s,l_{s}}^{\text{S}}}\left[\ln\left(1 + \gamma_{s,l_{s}}^{\text{S}}\right)\right]\right],
\end{equation}
that is provided in the following Theorem.
\begin{thm} \label{thm_4}
	The average rate for a typical SUE of seller MNO $s$ is
	\begin{equation}\label{r_n}
	\begin{aligned}
	\mathcal{R}_{s}^{\rm{S}}&=\sum_{l_{s}\in \mathcal{L}_{s}} \mathbb{E}_{P^{\rm{B}}_{l_{s}}} \left[\mathbb{E}_{{\gamma}_{s,l_{s}}^{\rm{S}}}\left[\ln\left(1 + {\gamma}_{s,l_{s}}^{\rm{S}}\right)\right]\right]\\
	&=\sum_{l_{s}\in \mathcal{L}_{s}} \mathbb{E}_{P^{\rm{B}}_{l_{s}}} \left[\int_{0}^{\infty}\mathbb{P}\left[\ln\left(1 + {\gamma}_{s,l_{s}}^{\rm{S}}\right) > t\right] {\rm{d}} t\right]\\
	&=\sum_{l_{s}\in \mathcal{L}_{s}}\int_{0}^{\infty}\int_{0}^{+\infty}\pi\lambda_{s}{P^{\rm{S}}}\,^{2/\alpha} \exp\left(-(e^{t}-1)\sigma^{2}z^{\alpha/2}\right) \\
	&\times\exp\left(- \pi\lambda_{s}{P^{\rm{S}}}\,^{2/\alpha}z\right) \mathcal{L}_{I_{s,l_{s}}^{\rm{S}}}\left((e^{t}-1)z^{\alpha/2}\right) {\rm{d}}z \;{\rm{d}}t.
	\end{aligned}
	\end{equation}
\end{thm}
\section{Numerical results}\label{Res}
In this section, we validate the analytical results and evaluate
the performance of our adopted power control strategy. To this end, we consider a cellular network with one seller MNO and one buyer MNO. The BSs and UEs of MNOs are spatially distributed according to independent PPPs inside a circular area of 500 meter radius. The summary of other simulation parameters is given in Table \ref{tab:parameters}.

Fig. \ref{demo2} illustrates the coverage probability for seller and buyer MNOs versus SINR threshold ($\beta$) for different values of $\lambda_{b}\;(b\in\mathcal{B})$. We can observe that our analytical expressions provided in Theorem \ref{thm_1} and \ref{thm_2} and simulations are perfectly matched. From Fig. \ref{demo2}, it can be seen that increasing the value of buyer MNO's BS intensity ($\lambda_{b},b\in\mathcal{B}$) results in improving the coverage probability of the buyer MNO, while bringing about a negligible impact on the seller MNO's coverage probability. It is worth noting that the seller MNO and buyer MNO can gain significant coverage probability for their UEs by selecting a suitable value of BS intensity.

Fig. \ref{demo4} demonstrates the coverage probability for seller and buyer MNOs versus SINR threshold ($\beta$) for different values of $\mu_{s}\;(s\in\mathcal{S})$. This figure shows that with increasing the value of seller MNO's UE intensity, the coverage probability of the seller MNO remains fixed. Nonetheless, the coverage probability of the buyer MNO decreases. The reason is that with increasing the UE intensity of the seller MNO, $P_{l_{s}}^{\text{B}}$ may decrease based on its definition in Lemma \ref{lemma_2} leading to a lower coverage probability for the buyer MNO. This claim has been confirmed in both simulation and analytical results in Fig. \ref{demo4}.
\begin{table}[t]
	\caption{Simulation parameters}
	\label{tab:parameters}
	\centering
	\renewcommand{\arraystretch}{1.4}
	\begin{tabular}{|l|l||l|l|}
		\cline{1-4}
		\textbf{Parameter}&\textbf{Value} &\textbf{Parameter}&\textbf{Value} \\
		\hline
		\hline
		$|\mathcal{S}|$& 1 & $|\mathcal{B}|$& 1\\
		\hline
		$\lambda_{s},s\in\mathcal{S}$&$8/(\pi\times{500}^{2})$ &
		$\lambda_{b},b\in\mathcal{B}$&$[8,16]/(\pi\times{500}^{2})$\\
		\hline
		$\mu_{s},s\in\mathcal{S}$&$[50,70]/(\pi\times{500}^{2})$&
		$\mu_{b},b\in\mathcal{B}$&$50/(\pi\times{500}^{2})$\\
		\hline
		$P^{\text{S}}$& 10 dBm &$\zeta_{s}, s\in\mathcal{S}$& -100 dBm\\
		\hline
		$\alpha$&5 &$\sigma^{2}$& -120 dBm\\
		\hline
	\end{tabular}
\end{table}
\begin{figure}[!t]
	\centering
	{\includegraphics[width=9.5 cm,height=6.5 cm]{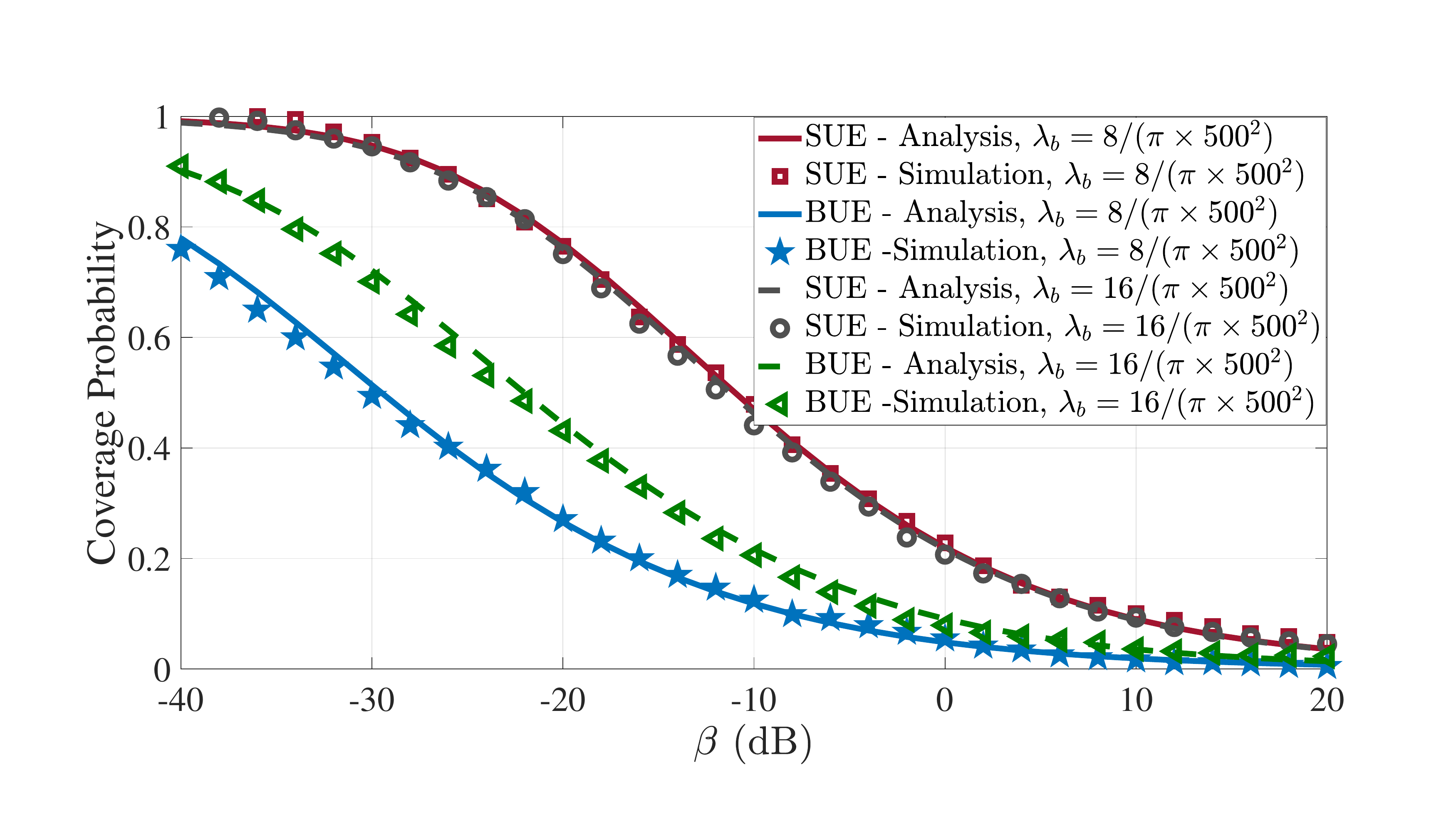}
		\caption{Coverage probability of seller and buyer MNOs versus SINR threshold ($\beta$) for different values of  $\lambda_{b}\;(b\in\mathcal{B})$.\label{demo2}}}
\end{figure}
\begin{figure}[!t]
	\centering
	{\includegraphics[width=9.5 cm,height=6.5 cm]{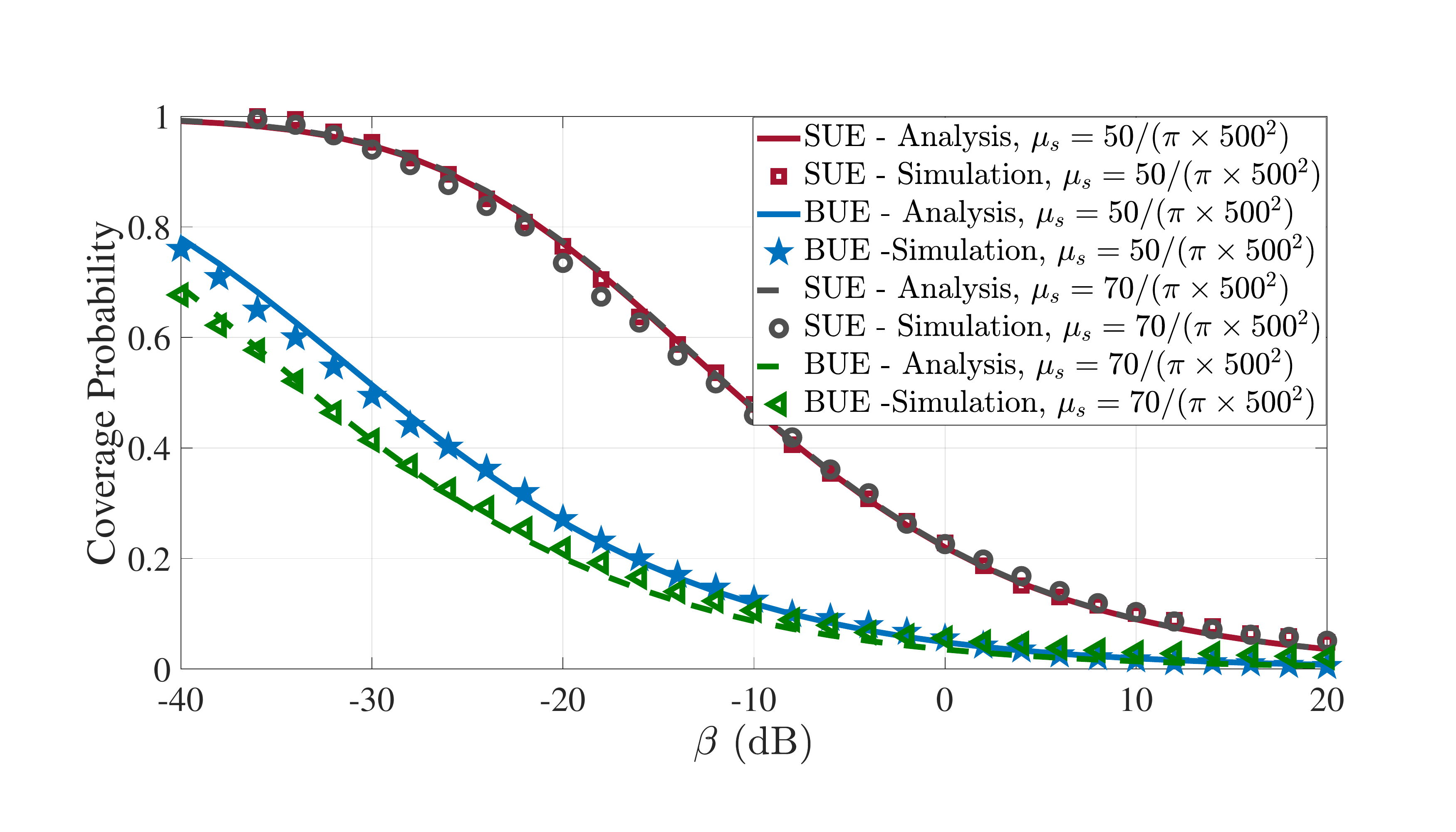}
		\caption{Coverage probability of seller and buyer MNOs versus SINR threshold ($\beta$) for different values of  $\mu_{s}\;(s\in\mathcal{S})$.\label{demo4}}}
\end{figure}
\begin{figure}[!t]
	\centering
	{\includegraphics[width=9.5 cm,height=6.5 cm]{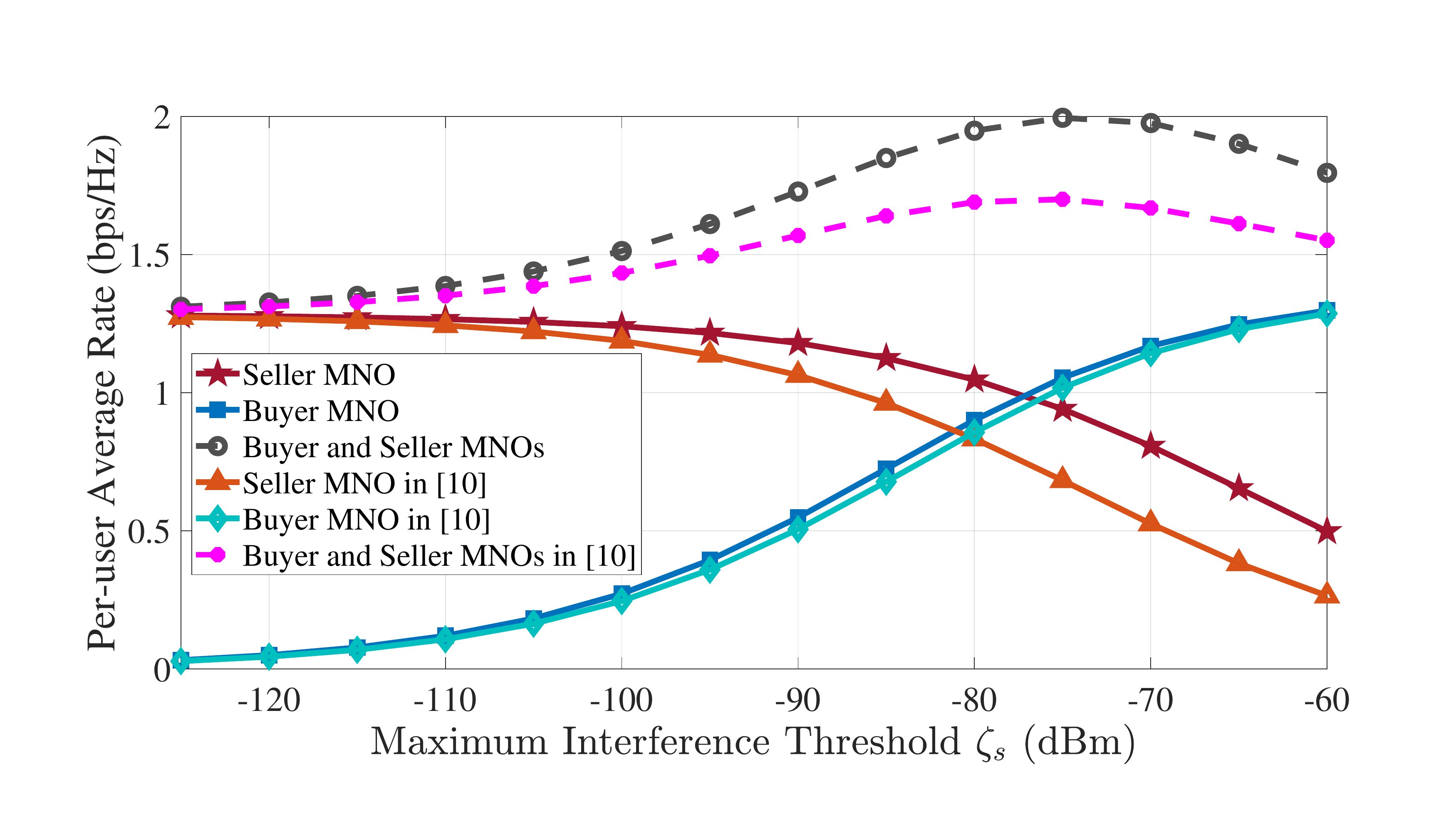}
		\caption{The per-user average rate of the seller and buyer MNOs, as well as the total sum-rate versus maximum interference threshold ($\zeta_{s}$).\label{demo5}}}
\end{figure}
The per-user average rate of both the seller and buyer MNOs, as well as their total sum-rate versus the maximum interference threshold ($\zeta_{s}$) is observed in Fig. \ref{demo5}. From this figure, we realize that increasing the maximum interference threshold ($\zeta_{s}$) leads to a decrease in the per-user average rate for the seller MNO and an increase in the per-user average rate of the buyer MNO. The reason is that when the maximum interference threshold increases, the transmit power of buyer MNO's BSs ($P_{l_{s}}^{\text{B}}$) based on \eqref{Power} can increase that causes to a higher per-user average rate. Since the transmit power of the seller's BSs is constant, when the maximum interference threshold ($\zeta_{s}$) increases, the seller MNO's per-user average rate decreases according to Theorem \ref{thm_4}. For a particular network setting, increasing the maximum interference threshold ($\zeta_{s}$) contributes to an increasing total sum-rate; nonetheless, after some points of $\zeta_{s}$, the total sum-rate decreases because the interference on SUEs is intensified. Another interesting observation is that there is a measurable value for the maximum interference threshold ($\zeta_{s}$) that maximizes the total sum-rate of MNOs. Besides, we compare our adopted power control strategy provided in subsection \ref{Power_Strategy} with the power control strategy introduced in \cite{Gupta2016}. In \cite{Gupta2016}, a system consisting of one seller and one buyer MNO is considered where the buyer MNO adjusts the transmit power of its BSs according to the distance of the nearest SUE to each of its BSs. From Fig. \ref{demo5}, we see that our adopted power control strategy outperforms the another strategy proposed in \cite{Gupta2016}. In fact, our adopted power control strategy adjusts the transmit power of buyer MNO's BSs in such a way that the interference imposed on all corresponding SUEs is below $\zeta_{s}$. Thus, the per-user average rate of seller MNO along with the total sum-rate improves significantly.
\section{Conclusion}\label{Con}
In this paper, we modeled a non-orthogonal spectrum sharing among multiple seller and multiple buyer MNOs employing stochastic geometry, where each seller MNO may lease multiple licensed sub-bands to buyer MNOs. Moreover, each buyer MNO could lease multiple licensed sub-bands from different seller MNOs. We assumed each MNO owns its BSs and UEs distributed based on two independent homogeneous PPPs. To reduce the interference on one seller MNO's UEs on one licensed sub-band, buyer MNOs had a limitation on the maximum interference imposed on the corresponding SUEs. After that, we analyzed the downlink coverage probability and the average rate of both seller and buyer networks. Finally, we validated our analysis with simulation and illustrated the performance of our adopted power control strategy via numerical results.

\end{document}